\documentclass[journal]{IEEEtran}
\usepackage{graphicx}
\usepackage{amsmath}
\usepackage{amssymb}
\usepackage{makecell}
\usepackage{pifont}
\usepackage{enumitem}
\usepackage{algorithm}
\usepackage{algorithmic}
\usepackage{multirow}
\usepackage{bigdelim}
\usepackage{fancyhdr}
\newcommand{\xmark}{\ding{55}}

\usepackage{dblfloatfix}
\usepackage{amsthm}
\usepackage{hyperref}

\newtheorem{theorem}{Theorem}

\usepackage[table]{xcolor}
\usepackage{soul}
\usepackage{hhline}
\usepackage{colortbl}

\newcommand{\lighttgray}{\rowcolor[gray]{.95}}
\newcommand{\lightgray}{\rowcolor[gray]{.90}}
\newcommand{\darkgray}{\rowcolor[gray]{.75}}

\newcommand{\rev}[1]{{\leavevmode\color{black}#1}}

\newlength{\Oldarrayrulewidth}
\newcommand{\HHHline}[2]{%
  \noalign{\global\setlength{\Oldarrayrulewidth}{\arrayrulewidth}}%
  \noalign{\global\setlength{\arrayrulewidth}{#1}}\hhline{#2}%
  \noalign{\global\setlength{\arrayrulewidth}{\Oldarrayrulewidth}}}

\begin{document}

\title{Nocturnal Seizure Detection Using \\ Off-the-Shelf WiFi}

\author{{Belal Korany,~\IEEEmembership{Student Member,~IEEE,}
        and Yasamin Mostofi,~\IEEEmembership{Fellow,~IEEE}}
\thanks{B. Korany, and Y. Mostofi are with the Department
of Electrical and Computer Engineering, University of California, Santa Barbara, CA, 93106 USA. E-mail: \{belalkorany, ymostofi\}@ece.ucsb.edu. 
}
}

\maketitle

\begin{abstract}
Detection of nocturnal seizures in epilepsy patients is essential, both for the quick management of the seizure complications, and for the assessment of the ongoing seizure treatment. Traditional seizure detection products (e.g., wearables), however, are either very costly, uncomfortable, or unreliable. In this paper, we then propose to utilize everyday WiFi signals for robust, fast, and non-invasive detection of nocturnal seizures. We first present a new and rigorous mathematical characterization for the spectral content/bandwidth of the WiFi signal, measured on a WiFi device placed near a sleeping patient, during different kinds of sleep motions: seizures, normal movements (e.g. posture adjustments), and breathing. Based on this mathematical modeling, we propose a novel pipeline for processing the received WiFi signals to robustly detect all nocturnal non-breathing movements, and then classify them into normal body movements or seizures. In order to validate this, we carry out extensive experiments in 7 different typical bedroom locations, where a set of 20 actors simulate the state of having seizures (total of 260 instances), as well as normal sleep movements (total of 410 instances). Our proposed system detects 93.85\% of the seizures with a mean response time of only 5.69 seconds since the onset of the seizure. Moreover, our proposed system achieves a probability of false alarm of only 0.0097, when classifying normal sleep movements. Overall, our new mathematical modeling and experimental results show the great potential the ubiquitous WiFi signals have for detecting nocturnal seizures, which can provide better support for epilepsy patients and their caregivers.
\end{abstract}


\maketitle

\section{Introduction}\label{sec:intro}
Epilepsy is a neurological disorder that causes a patient to have different kinds of seizures. It has gained a lot of attention in the public health domain since it is one of the most common neurological disorders, causing a large number of people to suffer from persistent health and socioeconomic issues. The World Health Organization (WHO) estimates that 50 million people around the world suffer from epilepsy, as of 2019 \cite{who_epilepsy}.  
Epilepsy is treated using different Anti-Epileptic Drugs (AEDs), depending on the specific type of seizure it is causing. Assessment of the ongoing seizure treatment requires the caregivers of the patient to continuously monitor and document the seizures (i.e., their frequency and duration). Seizures which take place during night sleep (\textbf{medically known as \textit{Nocturnal Seizures}}) then pose a higher risk for epilepsy patients, since they can go unobserved by the caregivers \cite{van2018nocturnal}. This necessitates the need for in-home seizure monitoring devices that can detect nocturnal seizures in epilepsy patients and alert their caregivers. 
The presence of a caregiver during a seizure is also very important so that they can help the patient, prevent them from falling, administer rescue medications (if necessary), and/or call for medical help if the seizure is lasting for too long. Moreover, patients who continue to have unattended nocturnal seizures have a higher risk of death due to the complications caused by the unattended seizures, a condition that is medically known as Sudden Unexpected Death in Epilepsy (SUDEP) \cite{van2018nocturnal}. SUDEP has been found to usually follow a specific type of seizures called \textit{tonic-clonic} seizures, which happens more frequently than other types during sleep \cite{lhatoo2016nonseizure}. Tonic-clonic seizures are characterized by a \textit{tonic} phase, in which the body muscles stiffen for a few seconds, followed by a \textit{clonic} phase, in which the body muscles rapidly and rhythmically jerk for 1-3 minutes \cite{jenssen2006long}.\footnote{This paper will focus on tonic-clonic seizures. Therefore, unless otherwise stated, we henceforth use the term "seizure" to refer to a tonic-clonic seizure.}

In order to detect nocturnal seizures, several products have been made available, such as smart watches \cite{embrace2}, smart mattresses \cite{mp5}, and cameras \cite{sami}. Smart watches measure the acceleration of the wrist to detect violent jerky movements, while smart mattresses measure the changes in the pressure on the mattress. However, the large cost of these products prohibits their widespread use. For instance, smart watches (e.g., Apple Watch and Embrace2 \cite{embrace2}), and the MP5 bed motion monitoring unit \cite{mp5} all cost more than \$250 per unit. Moreover, the comfort of patients and the reliability of detection of some of these products have been questioned by several studies, as we shall discuss in Sec.~\ref{sec:related_work}.

On the other hand,  Radio Frequency (RF) signals (e.g., WiFi) have become ubiquitous these days, due to the rapid growth of the number of wireless devices. These signals interact and bounce off of different objects in the environment, thereby carrying crucial information about them. Consequently, researchers in the RF sensing community have utilized RF signals to realize various applications, e.g. localization and tracking \cite{karanam2018magnitude}, imaging \cite{gonzalez2013cooperative}, health monitoring \cite{nandakumar2015contactless}, occupancy estimation \cite{depatla2018crowd},  activity recognition \cite{cai2020teaching}, and others. 

In this paper, we propose to utilize RF signals to detect nocturnal seizures in epilepsy patients. Using everyday RF signals, i.e. WiFi signals, for such a task has several advantages. First, it is an affordable solution when compared to the high cost of existing approaches. It is also contactless since it does not require the patient to wear any device or have units installed under their mattress. Moreover, an RF-based system, unlike cameras, does not require any lighting conditions to accurately achieve its task. In this paper, we then propose to use a pair of WiFi transceivers to detect nocturnal seizures. More specifically, we propose a \textit{robust}, \textit{fast}, and \textit{theoretically-driven} approach to process the WiFi Channel State Information (CSI) measured on a WiFi receiver device placed near a sleeping patient, in order to extract their motion information and decide whether the motion indicates a seizure or not. By "robust", we mean that our proposed framework has a very low probability of false alarm, i.e., it has a very low probability of declaring a seizure when there is none, while detecting all the seizures with a high probability. This is important as the sleeping person may have several normal body movements, such as pose adjustments, and they should not be classified as a seizure. By "fast", we mean that our system detects a seizure in a very short time since its onset, in order to alert the caregiver in a timely manner.  Finally, by "theoretically-driven", we mean that our proposed approach is backed by a new and rigorous mathematical characterization of the spectral content of the received signal during sleep-related movements: seizure, normal body movements, and breathing.

In our setup, a pair of WiFi transceivers (e.g., two laptops) are placed near the patient's bed. The WiFi receiver measures both the WiFi CSI squared magnitude signal and the phase difference between the antennas of the receiver (total of 3 antennas) for the purpose of seizure detection. We then propose a new mathematical characterization for the spectral content of the received WiFi signal during motions relevant to sleep, i.e., seizure, normal sleep movements, and breathing, and show how our spectral analysis can be used to design a new and robust seizure detection pipeline.  We next explicitly discuss the contributions of this paper.

\textbf{Statement of Contributions:}

\noindent 1. We develop a novel and rigorous mathematical model for the received CSI squared magnitude signal as well as the CSI antenna phase difference during different kinds of motions relevant to sleep: seizure, normal body movements, and breathing. More specifically, we first show that both the WiFi CSI squared magnitude and phase difference signals are frequency-modulated by the body motion. 
We then show our main theoretical contribution: to mathematically characterize the spectral content/bandwidth of the WiFi CSI signal during the aforementioned motions. Based on this new spectral analysis, we then show that the bandwidth of the received WiFi signal can be used to robustly and efficiently differentiate seizure events from normal sleep movements. 

\noindent 2. Based on our theoretical analysis, we propose a new pipeline for the detection of nocturnal seizures using WiFi CSI, which consists of the following 3 steps. First, our data pre-processing pipeline denoises the raw measured CSI and selects the least noisy data streams of different receiver antennas/subcarriers using our spectral analysis findings. Then, our event detection algorithm detects any kind of non-breathing motion, based on the spectral content of the denoised CSI. Finally, an event classification algorithm decides whether a detected event is a seizure or normal body movement, based on the bandwidth of the WiFi signal during the event.

\noindent 3. In order to validate our proposed framework, we carry out extensive experiments on 20 test subjects (5 females and 15 males) in 7 different locations of typical bedrooms, where the subjects act out seizures and normal sleep movements while we collect WiFi CSI data. In total, we collect 260 different seizure instances and 410 different normal non-breathing sleep movement instances. Our system was able to detect 93.85\% of the seizures with an average response time of 5.69 seconds since the onset of the seizure, which is much less than the state of the art, as we shall see in Sec. \ref{sec:results}. Moreover, in terms of false alarm rate (the probability that a normal sleep event is classified as seizure), our system had a false alarm probability of 0.0097, which indicates its robust performance.  We further study the impact of varying several different parameters (e.g., TX/RX positions) on the performance of our proposed system. Overall, our results establish that our proposed mathematically-motivated system is fast and robust and is also independent of person's pose/orientation.

As we shall see, our derivations can also contribute beyond seizure detection, in the general area of breathing-based RF sensing, since they show that a common assumption regarding the frequency content of the received signal during normal breathing is not always correct, explaining some of the unexplained observations in the corresponding literature.

\textbf{Remark 1:} In this paper, we use the term \textit{normal sleep events} to refer to normal non-breathing body movements during sleep, such as pose adjustments, stretching, scratching, coughing, sneezing, jerking, and others.

\section{Related Work}\label{sec:related_work}

To the best of our knowledge, this work is the first to use RF signals for seizure detection.  In this section, we summarize the state-of-the-art related to different aspects of our problem of interest. 

\subsection{WiFi-based Vital Signs Monitoring}\label{sec:breath_literature}
There has been a great body of work on utilizing wireless signals for vital signs monitoring, e.g. using high-bandwidth radar \cite{adib2015smart}, mmWave \cite{yang2016monitoring}, or WiFi. In this paper, we are interested in utilizing off-the-shelf WiFi devices for seizure detection.

Several papers have utilized the fine-grained WiFi CSI magnitude data for breathing rate and/or heart rate estimation \cite{liu2015contactless, wang2016human}. Other researchers utilized the CSI phase difference between receiver antennas to achieve the same task \cite{wang2017phasebeat, wang2017resbeat}. None of such RF-based existing work, however, is on seizure detection. Nevertheless, our findings can have a significant impact on such work for the following reason. All the existing WiFi CSI-based breathing rate estimation work assume that the bandwidth of the received CSI signal, in the vicinity of a person who is breathing normally, is the same as the breathing rate.  In order for us to develop a robust nocturnal seizure detection system, we also need to fundamentally understand and characterize the spectral content of normal breathing in this paper.  As we shall see, using our proposed rigorous mathematical analysis, the bandwidth of the WiFi signal caused by normal breathing is not necessarily the same as the breathing rate and can be higher. As such, this paper can contribute to the ongoing research that is using breathing signals for other health monitoring applications.  In fact, our mathematical analysis can immediately explain the observation made in \cite{wang2016human} that the quality of the breathing rate estimation, which was designed assuming the signal bandwidth is the same as the breathing rate, degrades at some locations relative to others.  Similarly, it can explain the unexplained frequency peaks that were observed in \cite{wang2020respiration} and were attributed to noise.

\subsection{Seizure Detection and Analysis}
In-home seizure detection is an important topic that has gained a lot of attention in the research community. Most seizure detection algorithms in the literature rely on the detection of the motion of the clonic phase of the tonic-clonic seizure via accelerometry \cite{nijsen2008accelerometry, velez2016tracking, kusmakar2018automated}, 
and/or video analysis \cite{kalitzin2012automatic, geertsema2018automated}. In accelerometry, a wearable accelerometer is attached to one or more of the patient's body parts, such as wrist, ankle, and/or chest. In addition to their high cost, wearable devices are usually not well tolerated by certain groups of patients, such as children and people with intellectual disabilities, who usually try to dislodge the devices \cite{geertsema2018automated}. Furthermore, the authors of \cite{beniczky2013detection} concluded that commercial wrist-worn watches have a seizure detection accuracy of 89.7\%, which is not very high. 
Video-based seizure detection has shown good detection accuracy of more than 95\%, but with a high false alarm rate of 0.78 events per night {\cite{geertsema2018automated, kalitzin2012automatic}. However, video-based detection requires a clear unobstructed view of the patient with good lighting conditions, which may not always be possible, and further invade the patient's privacy. 
Overall, an accurate, non-invasive, comfortable, and affordable way of detecting nocturnal seizures is lacking, which is the main motivation for this paper.

\subsection{Sleep Analysis}
People engage in different kinds of normal non-breathing movements during sleep, such as posture adjustments and limb jerks. An important aspect of our proposed system is to minimize false alarms by identifying such \textit{normal events} and distinguishing them from seizure events. Hence, we utilize some of the results of the sleep analysis medical literature in order to design our system. For instance, \cite{coussens2014movement, de1992sleep} study the duration and rate of normal events during sleep, using accelerometry and video analysis, concluding that these events typically happen at a rate of 3 events per hour, last for an average of 8 to 10 seconds, and can go up to 15 seconds. The authors of {\cite{walch2019sleep}} have published an online dataset of accelerometry data of 31 healthy adults during their sleep. We shall utilize such results/data in this paper for our spectral analysis of  WiFi signals during normal sleep.

\section{Signal Model} \label{sec:signal_model}
In this section, we develop a mathematical model for the received WiFi CSI in a general setting, an example of which is shown in Fig. \ref{fig:scenario}. More specifically, a person is lying down on a bed in any generic pose while a WiFi transmitter (Tx) emits wireless signals that are reflected off of the person's body and received by a WiFi receiver (Rx). We first derive closed-form expressions for the WiFi CSI squared magnitude and the WiFi CSI phase difference signals, during a generic motion pattern of the body, in this part. In Sec. \ref{sec:spectral_analysis}, we then use this model to provide a new and rigorous mathematical analysis of the spectral content of the WiFi signals during specific kinds of sleep motions relevant to this paper, i.e., breathing, seizures, and normal sleep movements.

\begin{figure}
    \centering
    \includegraphics[width=0.6\linewidth]{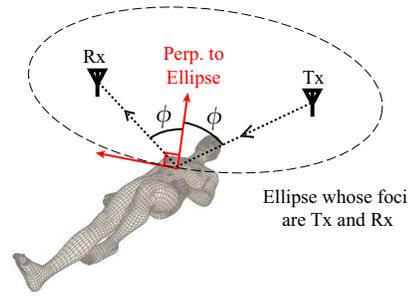}
    \caption{\small Illustration of the application scenario. A pair of WiFi transceivers collect WiFi CSI measurements while a person is sleeping, in order to analyze their sleep motions and detect if they are having a seizure. Note that our design does not assume or require that the person lies on their back, and they can be in any pose/orientation. Furthermore, the TX/RX can be in any configuration as well.}
    \label{fig:scenario}
\end{figure}

Let $c(t)$ denote the complex baseband received signal at the Rx, which can be decomposed into the direct path from the Tx to the Rx, and the reflected path off of the person's moving body. 
More specifically, $c(t)$ can be written as \cite{karanam2018magnitude}
\begin{equation}\label{eq_baseband}
    c(t) = \underbrace{\alpha_d e^{j\mu_d}}_{\substack{\text{direct path}}} + \underbrace{\alpha_r e^{j\left(\mu_r + \frac{2\pi}{\lambda} \psi \int \! v(t) dt \right)}}_{\text{reflected path}},
\end{equation}
where $\alpha_d$ and $\mu_d$ are the amplitude and phase of the direct path from the Tx to the Rx, $\alpha_r$ is the amplitude of the reflected path arriving at the Rx, $\mu_r$ is the phase of the reflected path at time $t=0$, $\psi=2\cos(\phi)$ is a scale parameter that depends on the location of the bed/person with respect to the Tx and Rx. Consider the ellipse whose foci are the Tx and Rx, which passes through the person's body, $\phi$ is then the angle between the line connecting the person to the Tx (or Rx) and the perpendicular line to this ellipse at the point that it passes through the person's body (see Fig.~\ref{fig:scenario}). $v(t)$ is the instantaneous speed component of the body motion along the perpendicular line to the ellipse, and $\lambda$ is the wavelength.

Note that the value of $\psi$ depends on the scene configuration, i.e., the relative location of the bed with respect to the Tx and Rx, and does not depend on the person's posture and orientation while sleeping. In other words, if the width of the bed is small as compared to the Tx-Rx distance, or the person does not drastically move from one side of the bed to the other (which is common in practice), the sleeping person's general location with respect to the Tx and the Rx does not drastically change, and hence, $\psi$ can be taken as a constant and can be calculated only once upon Tx-Rx placement. The person can change their pose/orientation several times, but those movements will not affect the value of $\psi$.

For simplicity of notation, we define $\beta=\frac{2\pi\psi}{\lambda}$, and $d(t)=\int \! v(t) dt$. Hence, the phase of the reflected path at the Rx becomes $\mu_r + \beta d(t)$.
Next, we derive closed-form expressions for the squared magnitude and phase of $c(t)$ to understand the information they carry about the body's motion.

\textbf{Remark 2:}
The static multipath in the environment does not affect this analysis since all the static multipath can be integrated into the first term of Eq. \ref{eq_baseband}. This indicates that the performance of the system is agnostic to the deployment environment. This observation will be further validated by our extensive experiments in several different locations and real-world scenarios, as we shall see in Sec. \ref{sec:results}.

\noindent\textbf{Squared Magnitude of $c(t)$:} The squared magnitude of $c(t)$ can be written, after a straightforward derivation, as follows,
\begin{equation}
    |c(t)|^2 = c(t) c^*(t) = \alpha_d^2 + \alpha_r^2 + A_m \cos \left(\beta d(t) + \Delta\mu_m\right),
\end{equation}
where $A_m = 2\alpha_d\alpha_r$, and $\Delta\mu_m=\mu_r-\mu_d$ is the difference between the initial phase of the reflected path and the phase of the direct path. Since the DC component of $|c(t)|^2$ does not carry any information about the motion of the body, we subtract the DC term (which can be easily implemented in practice) to have the following,

\begin{equation}\label{eq_squared_magnitude}
    s_m(t) = A_m\cos \left(\beta \int\!\! v(t) dt +  \Delta\mu_m\right).
\end{equation}

For the ease of discussion, we then refer to $s_m(t)$ as the squared magnitude signal in the rest of the paper.

\noindent\textbf{Phase of $c(t)$:}
Without loss of generality, we analyze the phase of the scaled signal $c'(t) = e^{-j\mu_d} c(t)/\alpha_d$. This scaling shifts the phase of $c(t)$ by a constant amount, preserving the time-varying behavior of the phase of $c(t)$ which carries the motion information of the body. Let $\theta(t)$ be the phase of $c'(t)$. It is easy to confirm that
\begin{equation}
    \theta(t) = \tan^{-1} \left( \frac{\frac{\alpha_r}{\alpha_d}\sin\left(\beta d(t) + \Delta\mu_m\right)}{1 + \frac{\alpha_r}{\alpha_d}\cos\left(\beta d(t) + \Delta\mu_m\right)} \right).
\end{equation}
Due to its longer length and the reflection loss at the body, we can assume that the amplitude of the reflected path is much less than that of the direct path, i.e. $\frac{\alpha_r}{\alpha_d} \ll 1$. In such a case, $\theta(t)$ can be approximated as
\begin{align} \label{eq_theta_approx1}
    \theta(t) &\approx \tan(\theta(t)) \nonumber\\
    % =\frac{\frac{\alpha_r}{\alpha_d}\sin\left(\beta d(t) + \Delta\mu_m\right)}{  1 + \frac{\alpha_r}{\alpha_d}\cos\left(\beta d(t) + \Delta\mu_m\right)} \nonumber\\
    &\approx \frac{\alpha_r}{\alpha_d}\sin\left(\beta d(t) + \Delta\mu_m\right) \left( 1 - \frac{\alpha_r}{\alpha_d}\cos\left(\beta d(t) + \Delta\mu_m\right) \right) \nonumber \\
    &= \frac{\alpha_r}{\alpha_d}\sin\left(\beta d(t) + \Delta\mu_m\right) - \frac{\alpha_r^2}{2\alpha_d^2}\sin\left(2\beta d(t) + 2\Delta\mu_m\right) \nonumber\\
    &\approx \frac{\alpha_r}{\alpha_d}\sin\left(\beta d(t) + \Delta\mu_m\right) ,
\end{align}
where the first order Taylor approximation $(1+x)^{-1} \approx 1-x$ for $x\ll 1$ is used in the second line to derive Eq. \ref{eq_theta_approx1}, since $\frac{\alpha_r}{\alpha_d} \ll 1$.

In practice, the phase measurements on off-the-shelf WiFi devices are corrupted by multiple sources of error, such as Carrier Frequency Offset (CFO) and Sampling Time Offset (STO), rendering these phase measurements unreliable \cite{zhuo2016identifying}. However, since different antennas of the same WiFi card share the same oscillator, those errors are common to all the antennas of the same card, and as such, the phase difference between two antennas of the same card carries stable phase information, as has been used in the literature. In this paper, we also rely on the phase difference between the antennas of one receiver WiFi card. Let $\theta_i(t)$ be the phase of the CSI at the $i$-th antenna of the Rx. The phase difference between the $i$-th and $j$-th receiver antennas can then be written as
\begin{equation} \label{eq_phase_diff}
    s_p(t) = \theta_i(t) - \theta_j(t) =  A_p\cos(\beta \int\!\!v(t) dt + \Delta\mu_p),
\end{equation}
where $A_p = 2(\alpha_r/\alpha_d) \sin\left(0.5(\Delta\mu_{m,i} - \Delta\mu_{m,j})\right)$, $\Delta\mu_p = 0.5(\Delta\mu_{m,i} + \Delta\mu_{m,j}) $, and $\Delta\mu_{m,i}$ and $\Delta\mu_{m,j}$ are the values of $\Delta\mu_m$ at the $i$-th and $j$-th receiver antennas, respectively.

Eq. \ref{eq_phase_diff} shows that as long as $\Delta\mu_{m,i} \ne \Delta\mu_{m,j}$ (which depends on the direct and reflected path lengths to the receiver antennas as well as the wavelength), the phase difference between the receiver antennas has a similar structure, in terms of the information it carries about the body movements, as the squared magnitude of the received signals (Eq. \ref{eq_squared_magnitude}).

\textbf{Body Acting as an FM Radio:} Frequency Modulation (FM) is a classic analog transmission technique, introduced in 1902 \cite{tucker1970invention}, to ensure robust transmissions for radio applications. A typical FM transmitted signal will have the form $\cos(2\pi f_ct +k_\text{f} \int \!m(t) dt)$, where $m(t)$ is the signal of interest to be transmitted, $f_c$ is the carrier frequency, and $k_\text{f}$ is the modulation index constant.  As can be seen, both the squared magnitude signal of Eq.~\ref{eq_squared_magnitude} and the phase difference signal of Eq. \ref{eq_phase_diff} can be interpreted as FM signals, in which $v(t)$ is the modulating signal and $f_c = 0$.  In other words, the moving body part (e.g., the chest) can be thought of as modulating the body motion into an FM signal that is then received by the WiFi receiver.  This way of interpretation allows us to delve into the classic mathematical analysis of FM signals for our system design, as we shall see in the next section.
However, one difference with a typical FM signal is the existence of the $\Delta\mu_m$ term in Eq. \ref{eq_squared_magnitude} (or $\Delta\mu_p$ in Eq. \ref{eq_phase_diff}). We shall see the impact of such a term in the spectral analysis of the next section.

\section{Spectral analysis of the Received Signal} \label{sec:spectral_analysis}

In this section, we analyze the received squared magnitude signal (or, equivalently, the phase difference signal) of Sec. \ref{sec:signal_model}, for different kinds of nocturnal body movements: breathing, seizures, and normal sleep events (e.g., posture shifts, moving limbs, etc.). More specifically, we develop our first major contribution: \textit{ \textbf{to mathematically characterize the spectral content/bandwidth of the received signal for each of the aforementioned three types of motions}}. We shall see that, due to the different body motion characteristics during a seizure as compared to normal sleep events, the spectral content of the received signals can be used to design a robust nocturnal seizure detection system, as we shall see in Sec. \ref{sec:system_description}.

Let $y(t)=A\cos(\beta \int\!\!v(t) dt+\Delta\mu)$ represent a general form for either the squared magnitude signal of Eq. \ref{eq_squared_magnitude} or the phase difference signal of Eq. \ref{eq_phase_diff}.  First, assume that $v(t)$ is a sinusoidal signal of the form $v(t)=v_{\text{max}} \cos(\omega_o t)$. This assumption applies to both the seizure and respiration cases.  The following characterizes the Fourier response of $y(t)$.

\begin{theorem}\label{theorem_spectrum}
Consider the signal $y(t) = A \cos(\beta \int\!\!v(t) dt + \Delta\mu)$ with a sinusoidal speed signal of $v(t) = v_{\text{max}} \cos(\omega_o t)$. The spectrum of this signal, i.e., its Fourier transform, can be written as follows,
\begin{align}\label{eq_theorem}
    Y(f) &= A\cos(\Delta\mu) \sum_{\substack{n \text{ even} \\ n \ge 0}} J_n(\beta') \left( \delta(f-nf_o) + \delta(f+nf_o) \right) \nonumber \\
    & +j A \sin(\Delta\mu) \sum_{\substack{n \text{ odd} \\ n > 0}} J_n(\beta') \left( \delta(f-nf_o) + \delta(f+nf_o) \right),
\end{align}
where $J_n(.)$ is the $n$-th order Bessel function, $\beta'=\beta v_\text{max}/\omega_o$, $\delta(.)$ is the Dirac-Delta function, and $f_o = \omega_o/2\pi$ is the fundamental frequency of $v(t)$.
\end{theorem}

% \begin{proof}
% See Appendix A.
% \end{proof}

\begin{proof}
If $v(t) = v_{\text{max}}\cos(\omega_ot)$, then $y(t)$ becomes
\begin{align}\label{eq_appendix_1}
    y(t) &= A\cos(\Delta\mu)\cos(\beta' \sin(\omega_o t)) \nonumber \\ &\qquad\qquad\qquad- A\sin(\Delta\mu)\sin(\beta' \sin(\omega_o t)) \nonumber \\
    &= A\cos(\Delta\mu)\mathcal{R}\left\{e^{j\beta' \sin(\omega_o t)}\right\} \nonumber \\ 
    &\qquad\qquad\qquad - A\sin(\Delta\mu)\mathcal{I}\left\{e^{j\beta' \sin(\omega_o t)}\right\}
\end{align}
where $\beta' = \beta v_\text{max}/\omega_o$, $\mathcal{R}\{.\}$ is the real part of the argument, and $\mathcal{I}\{.\}$ is the imaginary part of the argument. The exponential term $e^{j\beta'\sin(\omega_ot)}$ is periodic with a period $2\pi/\omega_o$, and can be expanded by its Fourier Series as \cite{lathi1998modern}
\begin{equation}\label{eq_fourier_bessel}
    e^{j\beta'\sin(\omega_ot)} = \sum_{n=-\infty}^\infty J_n(\beta')e^{jn\omega_ot},
\end{equation}
where $J_n(.)$ is the $n$-th order Bessel function. By substituting Eq. \ref{eq_fourier_bessel} into Eq. \ref{eq_appendix_1},  we get
\begin{align*}
    y(t) =
     \sum_{n=-\infty}^\infty A\,J_n(\beta')&\left( \cos(\Delta\mu)\cos(n\omega_ot)\right. \\ &\left.\qquad- \sin(\Delta\mu) \sin(n\omega_ot) \right).
\end{align*}
By making use of the fact that $J_{-n}(x) = (-1)^nJ_n(x)$, $y(t)$ can be written as
\begin{align*}
    y(t) =  &2A\cos(\Delta\mu) \sum_{\substack{n\text{ even}\\n\ge0}}J_n(\beta')\cos(n\omega_ot) \\ &- 2A \sin(\Delta\mu) \sum_{\substack{n\text{ odd}\\n>0}}J_n(\beta')\sin(n\omega_ot).
\end{align*}
By taking the Fourier transform of $y(t)$, we get Eq. \ref{eq_theorem}.
\end{proof}

Theorem \ref{theorem_spectrum} states that the spectrum of $y(t)$ consists of an infinite number of deltas, located at the fundamental frequency of $v(t)$ and its harmonics.  We next characterize the bandwidth of this signal. In order to do so, we need to find the frequency point after which the power of the subsequent delta functions has become negligible, as compared to the earlier terms.

\begin{theorem}\label{theorem_bw}
The bandwidth of $y(t)$ can be characterized as follows, for $\beta' \ge 1$, 
\begin{equation*}
    \left. BW \right\rvert_{\beta'\ge 1} = (\beta'+1)f_o = \psi v_\text{max}/\lambda+f_o,
\end{equation*}
where $f_o$ is the fundamental frequency of $v(t)$. Moreover, for $\beta'<1$, the bandwidth of $y(t)$ is best characterized as follows,
\begin{equation*}
    \left. BW \right\rvert_{\beta'<1} = 2f_o.
\end{equation*}
\end{theorem}
\begin{proof}
It is well established in the literature that $J_n(\beta')$ is negligible for $n > \beta'+1$ \cite{lathi1998modern}. By applying this to  Eq. \ref{eq_theorem}, we can then estimate the bandwidth as follows: $BW=(\beta'+1)f_o$ for $\beta'\ge 1$ since some even and odd terms are both present for $\beta' \ge 1$ and terms can be compared accordingly within each even and odd groups. When $\beta'< 1$, however, the previous result implies that the term $n=1$ is the only dominating term in the spectrum of $y(t)$. However, due to the different scaling factors of the even and odd terms in Eq.~\ref{eq_theorem}, there could exist cases (e.g., small $\Delta\mu$) where the term corresponding to $n=1$ is suppressed by the $\sin(\Delta\mu)$ factor. In such cases, even though $J_2(\beta')$ is small as compared to $J_1(\beta')$, $\cos(\Delta\mu)J_2(\beta')$ can be comparable or larger than $\sin(\Delta\mu)J_1(\beta')$. Higher order terms can always be neglected with respect to the first two terms. Hence, the bandwidth of $y(t)$ for the case of $\beta'<1$ is $2f_o$.
\end{proof}

\textbf{Remark 3:} In his seminal paper of {\cite{carson1922notes}}, J. Carson was the first to theoretically characterize the bandwidth of an FM signal and show that it can be larger than the bandwidth of the modulating signal. Carson has shown that his bandwidth rule is exact for sinusoidal modulating signals, but can be generalized to approximate the bandwidth for general non-sinusoidal modulating signals as well. As mentioned earlier, our received signal $y(t)$ has a close resemblance to an FM signal, except for the $\Delta\mu$ terms. As such, our bandwidth analysis has some resemblance to Carson's  derivations except for the impact of $\Delta\mu$. Following a similar argument to Carson's, we will then also use Theorem {\ref{theorem_bw}} to approximate the bandwidth of $y(t)$ when $v(t)$ is a general non-sinusoidal signal in the next section. In such a case, $f_o$ would denote the bandwidth of the signal $v(t)$.

Theorem \ref{theorem_bw} states that the bandwidth of $y(t)$ depends on motion parameters such as $v_\text{max}$ and $f_o$ (or the bandwidth of $v(t)$ for non-sinusoidal signals).  We next utilize Theorem \ref{theorem_bw} to estimate the bandwidth of the WiFi CSI \footnote{Henceforth, the bandwidth of WiFi CSI means either the squared magnitude or phase difference, since they both have the same generic form $y(t)$.} during three specific kinds of sleep-related motions: breathing, seizure, and normal sleep movements.

\subsection{CSI Bandwidth During Breathing}\label{sec_breathing_BW}
A sleeping person's chest volume expands and shrinks during the inhalation and exhalation phases of respiration. It is established in the literature that the instantaneous chest speed, i.e., $v(t)$ of Sec.~\ref{sec:signal_model}, can be approximated by a sinusoid of frequency $f_{o,\text{br}}$, where $f_{o,\text{br}}$ is the number of breathing cycles per second \cite{zhang2017wicare}. As such, Eq.~\ref{eq_theorem} can describe the spectrum of the WiFi signal during breathing.

In order to characterize the bandwidth for the case of normal breathing, we need to estimate $\beta'_\text{br} =\frac{\psi v_\text{max,br}}{\lambda f_{o,\text{br}}}$, where $f_{o,\text{br}}$ is the breathing rate of the person, which is typically in the range of 0.2 to 0.3~Hz~\cite{barrett2019ganong}. By integrating $v(t)$, it can be easily confirmed that $v_\text{max,br}/2\pi f_{o,\text{br}}$ is equal to the maximum chest displacement during respiration, which has been reported in the literature to be around $5~$mm \cite{wang2016human}. This results in $\beta'_\text{br} \approx 0.55$ when using WiFi channel 48, which has a carrier frequency $f_c = 5.24~$GHz, and $\psi = 1$.\footnote{In our experiments, we use WiFi channel 48 ($f_c = 5.24~$GHz) in a setup in which $\psi\approx 1$ (see Sec. {\ref{sec:experimental_validation}} for the detailed scene configuration). Extension to different values of $\psi$ is straightforward, as we shall discuss in Sec. \ref{sec:psi} where we show experiments with different $\psi$s. Hence, we set $\lambda = 5.72~$cm and $\psi=1$ for our numerical calculations in the rest of the paper up to Sec. \ref{sec:psi}.} By using Theorem \ref{theorem_bw}, we can then estimate the bandwidth of the received WiFi CSI during normal breathing as $BW_\text{br} = 2f_{o,\text{br}}$. Note that if the maximum chest displacement is not along the perpendicular line to the ellipse whose foci are the Tx and Rx (see Fig.~{\ref{fig:scenario}}), e.g., if the person is in a different pose, the chest speed will have a smaller velocity component along that line (i.e. smaller $v_\text{max,br}$). In such a case, the value of $\beta'$ will be even smaller and thus still less than 0.55. Thus, according to Theorem {\ref{theorem_bw}}, the bandwidth still remains $BW_\text{br} = 2f_{o,\text{br}}$ for all the cases.

    \begin{figure}
    \centering
    \includegraphics[width=.85\linewidth]{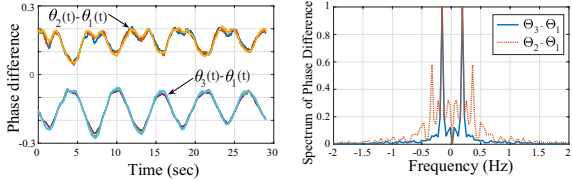}
    \caption{\small (Left) CSI phase difference at multiple subcarriers of the Rx and (right) the spectra of the CSI phase difference signals showing different spectral content for different antennas.}
    \label{fig:new_signal_model_example}
\end{figure}

It is worth noting that the previous literature on breathing monitoring using WiFi signals (either magnitude \cite{liu2015tracking, liu2015contactless,zhang2017wicare} or phase difference \cite{wang2017phasebeat}) assume that the received WiFi signal rises and falls with the same frequency as the rise and fall (inhalation and exhalation) of the chest during the breathing process. Hence, they assume that the received signal has the same spectral content/bandwidth as that of the physical chest motion (i.e. they take $BW_\text{br}$ to be $f_{o,\text{br}}$). 
However, Theorem \ref{theorem_bw} shows that the received WiFi signals can have a spectral content that is different from the physical breathing rate, depending on the value of $\Delta\mu$, with a maximum bandwidth of $2f_{o,\text{br}}$. To see this in effect, 
Fig.~\ref{fig:new_signal_model_example}~ shows the phase difference of the measured WiFi signals between the Rx antennas in a sample experiment, where a person was breathing with a frequency of $f_{o,\text{br}} \approx 0.18$ Hz. It can be seen that while the measured phase difference between antennas 3 and 1 of the Rx has a sinusoid-like pattern similar to that of the breathing motion, the phase difference between antennas 2 and 1 (due to having a different $\Delta\mu$) is experiencing a different pattern, which has a strong frequency component at $2f_{o,\text{br}}$.

\subsection{CSI Bandwidth During Seizures}
As described earlier, a tonic-clonic seizure consists of a \textit{tonic} phase, in which the body muscles stiffen for a few seconds, immediately followed by a \textit{clonic} phase, which is a strong, fast, and repeated stiffening and relaxing of the body muscles that can last for 1 to 3 minutes \cite{jenssen2006long}. Several medical studies have been conducted to analyze body motion during a tonic-clonic seizure through data obtained by accelerometry.  
These studies have found that during the clonic phase of a tonic-clonic seizure, the body muscles rhythmically stiffen and relax with a frequency $f_{o,\text{sz}}$ between $1.5$ and $5$ Hz \cite{quiroga2002frequency,luders1998semiological,nijsen2008accelerometry}, thus making a sinusoid a good approximation for $v(t)$. Therefore, Eq. \ref{eq_theorem} also characterizes the frequency spectrum of the WiFi CSI during a seizure.  In order to find the value of the parameter $v_\text{max}$, and thus $\beta'$, we have looked extensively into the medical literature on seizures.
Several papers have found that the maximum acceleration, $a_\text{max}$, of the body parts during a tonic-clonic seizure typically exceeds $15$~m/s$^2$  \cite{velez2016tracking,kusmakar2018automated}. Since $v(t)$ is sinusoidal, then $v_\text{max,sz}=\frac{a_\text{max}}{2\pi f_\text{sz}}$, and a lower bound for the value of $v_\text{max,sz}$ can be calculated as $v_\text{max,sz} \ge \frac{15}{2\pi\times5}=0.48$~m/s.

Based on the aforementioned seizure motion parameters, one can estimate a lower bound for the bandwidth of the WiFi CSI  during a seizure using Theorem \ref{theorem_bw} as  $BW_\text{sz} = (\beta'_\text{sz}+1)f_\text{sz} = \frac{\psi v_\text{max,sz}}{\lambda} + f_{o,\text{sz}}$. More specifically, by using WiFi channel 48 and $\psi=1$, $v_\text{max,sz} = 0.48$~m/s and $f_{o,\text{sz}} = 1.5~$Hz, a lower bound for the bandwidth of the WiFi signal during the seizure is estimated as $BW_\text{sz} \ge 9.9$~Hz. Note that the aforementioned characterization of the CSI bandwidth during a seizure assumes that the motion of at least one body part is aligned with (or has a strong component on) the perpendicular line to the Tx-Rx ellipse of Fig.~{\ref{fig:scenario}}. This assumption is practical since the uncontrolled muscle jerks during the seizure result in the body parts moving randomly in all different directions. Moreover, it has been shown in the medical literature that a patient's body posture can change to many different positions during a seizure {\cite{mahr2020prone}}. Therefore, there will at least be one body part whose motion direction is aligned with the perpendicular line to the Tx/Rx ellipse.

It is worth stressing that the traditional assumption that the WiFi signal rises and falls with the same frequency of the body motion will result in a bandwidth estimation of $BW_\text{sz} = f_{o,\text{sz}}$, which is far off from the true bandwidth during a seizure.

\rowcolors{2}{white}{gray!09}
\begin{table}[t!]
	\caption{\small Motion parameters and the corresponding bandwidth for 3 kinds of sleep movements.}
\label{table_summary_parameters}
	\centering
	\scriptsize
	\begin{tabular}[h]{m{2.6cm} | c | c | c} 
\Xhline{2\arrayrulewidth}
\darkgray
		   &  &  & \textbf{Normal} \\ 
     \darkgray    & \multirow{-2}{*}{\textbf{Breathing}} & \multirow{-2}{*}{\textbf{Seizure}} & \textbf{Event} \\

\Xhline{2\arrayrulewidth}
		
 %Duration (sec) & N.A. & 60-180 & 8-10.6 \\ 
 %\hline
 Motion Frequency (Hz) & $f_{o,\text{br}}\!=~$0.2-0.3$\!$  & $f_{o,\text{sz}}\!=$1.5-5$\!$  & $f_{o,\text{nm}}=~$2 \\
 \hline
 $v_\text{max}$ (m/s) & $\le$0.01 & $\ge$0.48  & $\le$0.33 \\
 \hline
 BW of WiFi signal (Hz) & $BW_\text{br}=2f_{o,\text{br}}$ & $BW_\text{sz}\ge~$9.9 & $BW_\text{nm}\!\le~$7.8$\!$\\

\Xhline{2\arrayrulewidth}

	\end{tabular}

\end{table}

\subsection{CSI Bandwidth During Normal Sleep Events}
We next delve into the medical literature on sleep motion analysis in order to characterize the parameters relevant for signal bandwidth characterization during normal sleep events, such as position adjustments and jerking in limbs, which people tend to make during different stages of sleep.
It is found that these normal sleep  events occur at an average rate of 3 events per hour \cite{de1992sleep}, and can last for up to 15 seconds each \cite{coussens2014movement}. Furthermore, other studies have performed  time-frequency analysis of the accelerometry data of normal sleep and established that most of the power of normal sleep event signals (e.g. $v(t)$) is concentrated below $f_{o,\text{nm}}=2$~Hz \cite{nijsen2010time}. While $v(t)$ is non-sinusoidal, and no exact closed-form expression exists for the spectrum of the CSI signals for a general $v(t)$, Theorem \ref{theorem_bw} can still be used to approximate the bandwidth of the WiFi CSI, as discussed in Remark 3.
    
In order to calculate an upper bound for the WiFi CSI bandwidth in case of normal sleep events, we focus on wrist movements during sleep, which can have higher speeds due to its relatively lower mass as compared to other body parts.
We utilize the online dataset published by the authors of \cite{walch2019sleep} for the accelerometry data of 31 adults during their sleep, collected from wrist-worn Apple watches. By integrating this acceleration data over time, we get the instantaneous speeds of the wrist during sleep. 
We then use the 99-th percentile value of the speeds calculated from the dataset, which is found to be 0.33 m/s, as an estimate for the maximum possible speed of body parts during normal sleep events.\footnote{Larger speed values are only recorded when quick jerky limb motions take place. Such events are easily identifiable and differentiable from seizures, since they typically last for less than 400~ms {\cite{luders1998semiological}}.} To estimate an upper bound for $v_\text{max}$ during normal sleep events, we assume that the body part with the fastest motion is aligned with the perpendicular line to the Tx-Rx ellipse of Fig. {\ref{fig:scenario}}. Hence, $v_\text{max,nm} \le 0.33~$m/s. Then, Theorem {\ref{theorem_bw}} estimates the bandwidth of the WiFi signals during a normal sleep event as $BW_\text{nm}= \frac{\psi v_\text{max,nm}}{\lambda} + f_{o,\text{nm}}$, where $f_{o,\text{nm}}$ denotes the bandwidth of the modulating signal $v(t)$. At WiFi channel 48 and $\psi=1$, an upper bound of this bandwidth will then be $BW_\text{nm}~\le~\frac{0.33}{\lambda}~+~2~=~7.8$~Hz.

Table \ref{table_summary_parameters} summarizes the results of our WiFi CSI bandwidth analysis during the three considered nocturnal movements: breathing, seizure, and normal sleep events. It can be seen from the table that the bandwidth of the WiFi signal during a movement can be used as a distinguishing feature that differentiates seizures from normal sleep events. We make use of this observation to design a robust nocturnal seizure detection system in the next section.

\begin{figure}
    \centering
    \includegraphics[width=.75\linewidth]{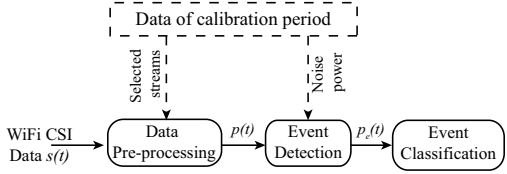}
    \caption{\small Block diagram of the proposed WiFi CSI-based nocturnal seizure detection system. The data pre-processing and event detection blocks utilize the derived WiFi CSI bandwidth during breathing ($BW_\text{br}$). The event classification module then utilizes the derived WiFi CSI bandwidth during both seizure ($BW_\text{sz}$) and normal sleep movements ($BW_\text{nm}$).}
    \label{fig:flowchart}
\end{figure}

\section{System Description} \label{sec:system_description}
In this section, we describe our proposed framework for nocturnal seizure detection using WiFi CSI signals based on the mathematical analysis of Sec. \ref{sec:spectral_analysis}. Fig. \ref{fig:flowchart} shows the block diagram of our proposed system, which starts by pre-processing the WiFi CSI input data to denoise the measured CSI signal and extract the part that carries the information about the human motion. Then, the denoised signal is passed to an event detection module, which decides whether the person is moving or is staying still. In case a movement event is detected (other than breathing), the CSI data during the event is then forwarded to an event classification module, which determines whether this event is a normal sleep event or a seizure. In the latter case, the system alarms the caregiver to take the necessary action. We next describe each of these components in details.

\begin{figure*}
    \centering
    \includegraphics[width=.72\linewidth]{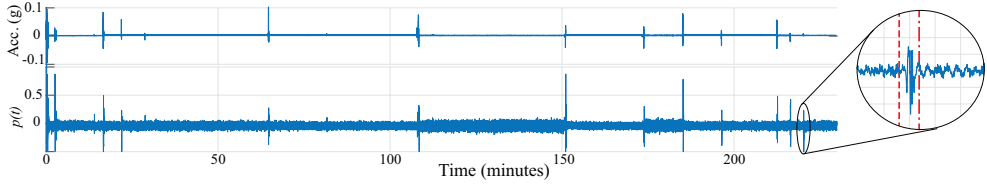}
    \caption{\small (Top) Sample output of the accelerometer attached to the arm of a subject during 4 hours of overnight sleep. (Bottom) The  PCA-denoised stream $p(t)$ of the WiFi data collected during the same time period. The dashed red line in the zoomed-in part shows the start of a sample event that is detected by our event detection module, while the dashed-dotted red line indicates its end.}
    \label{fig:example_overnight}
\end{figure*}

\subsection{Data pre-processing}\label{sec:pre-processing}
As discussed in Sec. \ref{sec:signal_model}, we utilize both the CSI squared magnitude and phase difference since they both carry crucial information about the body motion.
In this paper, we consider off-the-shelf WiFi devices that can be used to extract the complex WiFi CSI information, e.g. Intel 5300 or Atheros AR9580 WiFi cards. In any of these devices, the receiver has $N_R$ receiver antennas, which measure the WiFi CSI information on $N_{sc}$ subcarriers.  Therefore, we extract a total of $N_R\times N_{sc}$ CSI squared magnitude streams, and $(N_R-1)\times N_{sc}$ phase difference streams (i.e. the phase difference between each antenna and antenna 1, for all the $N_{sc}$ subcarriers). In total, we get $N_D = (2N_R-1)\times N_{sc}$ data streams that can be used to extract the motion information. The Intel 5300 WiFi card, for instance, has $N_R=3$ receiver antennas and $N_{sc}=30$ subcarriers, resulting in a total of $N_D = 150$ data streams carrying the motion information of the body. We next show how we process these $N_D$ data streams to extract the informative part about the body motion.

\textit{Outlier Removal:} We use the Hampel identifier \cite{liu2015contactless} to remove the sudden and very short abrupt changes that happen in the data streams due to hardware imperfections \cite{liu2015contactless}.

\textit{Stream Selection:} 
Different subcarriers on the same Rx antenna have different carrier frequencies (or wavelengths), and consequently, they undergo different levels of fading, making some subcarriers noisier than others.
To enhance the system's robustness, it is then important to select only the most informative/least noisy data streams to be subsequently used in the rest of the seizure detection algorithm. In order to do so, we use the data of a short \textit{calibration period}, in which the sleeping person is only breathing and not doing any movements or having a seizure. This one-time calibration can be easily administered by a caregiver prior to system deployment, and recalibration can be done as needed.

The stream selection algorithm works as follows. Since the calibration period is known to have only breathing motion, the CSI data contains frequency components only in the band $f\le BW_\text{br}$, where $BW_\text{br}$ is the maximum bandwidth for WiFi CSI during breathing, which we have shown in Sec. \ref{sec:spectral_analysis} to be  $2f_{o,\text{br}}$. Any frequency content above $BW_\text{br}$ is thus due to noise. Hence, given all the data streams in a calibration window of duration $T_\text{cal}$, we calculate the Signal-to-Noise Ratio (SNR) of the $i$-th data stream as follows
\begin{equation}
    SNR_i =  \sum\limits_{0<f\leq BW_\text{br}} \!\!\!\! S_i(f) \,\,\,\, \bigg/\,\, \sum\limits_{f>BW_\text{br}} \!\!\!S_i(f) ,
\end{equation}
where $S_i(f) = |\sum\limits_t s_i(t)e^{-j2\pi ft} |^2$, $s_i(t)$ is the $i$-th data stream, $BW_\text{br} =2f_{o,\text{br}}$, and $f_{o,\text{br}}$ is the maximum normal breathing frequency, which is equal to 0.3 Hz in adults.

We then select the $K$ data streams with the highest SNRs from the calibration data, and use only this set of streams in the operation phase until after a major event happens (for instance, a seizure). The system can then re-calibrate by processing all the data streams again and re-selecting the new top $K$ streams in terms of SNR in the new person's pose/orientation. For the implementation of our system (see details in Sec. \ref{sec:experimental_validation}), we set $T_\text{cal}~=~13$~sec and $K=15$.

\textit{PCA denoising:} After extracting the set of the best $K$ data streams, we further denoise these streams during operation phase using Principal Component Analysis (PCA) as described in \cite{wang2015understanding}. More specifically, we extract the first principal component $p(t)$ of the data, which carries the motion information since it is common to all the data streams, while the noise is distributed among all the different principal components \cite{wang2015understanding}. 

In order to show the performance of the preprocessing module, we conduct an overnight sleep  experiment, where WiFi transceivers are placed on both sides of a bed on which a subject sleeps. An accelerometer is attached to the upper right arm of the subject to collect ground truth sleep motion data. Fig. \ref{fig:example_overnight} shows a 4-hour snippet of the processed WiFi data $p(t)$ as well as the accelerometer output during the same period. A 13-second calibration period is chosen right after the subject goes to sleep and the selected streams are then used for the rest of the night. It can be clearly seen that the preprocessed WiFi data $p(t)$ carries the same motion information as the accelerometer. In the right part of the figure, we zoom in to one of the movements, where the breathing signal, as well as the motion, can be clearly seen in the WiFi data.

\subsection{Event detection}
As described in the previous section, the data pre-processing module outputs a signal $p(t)$ which is the denoised version of the CSI measurements at the receiver. This signal is then fed to an \textit{event} detection module. By "event", we mean the state of the sleeping person engaging in any kind of \textit{non-breathing} movement. More specifically, the movement can be normal sleep events, e.g., posture adjustments, or abnormal, e.g. a seizure. The nature of the event (whether it is normal or abnormal) will be decided in a later stage, which we shall describe in Sec. \ref{sec:event_classification}.

\begin{figure*}
\centering
  \includegraphics[width=0.75\linewidth]{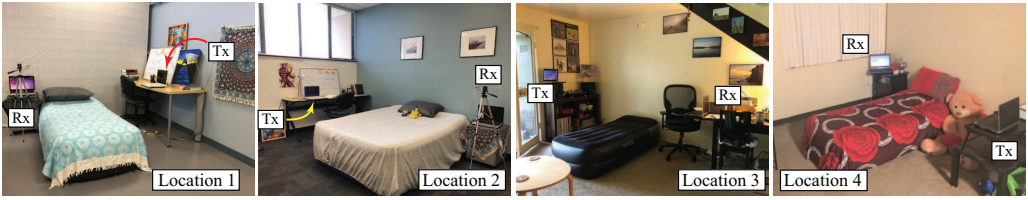}
  \caption{\small We tested our proposed approach in 7 different locations. Four sample locations are shown here.}
  \label{fig:locations}
\end{figure*}

In order to detect an event in the signal $p(t)$, we use a moving window of duration $T_\text{win}^\text{ED}$. If the person was only breathing during an instance of the moving window, the signal $p(t)$ during that window will have a frequency spectrum that is concentrated below $BW_\text{br}$, as discussed in Sec. \ref{sec:pre-processing}. On the other hand, if the person engages in any type of non-breathing movement, the signal $p(t)$ within the time window will have non-negligible frequency content above $BW_\text{br}$. Therefore, we can utilize the energy content of the spectrum of $p(t)$ above the frequency $BW_\text{br}$ to indicate the presence of an event. More specifically, let $\mathcal{H}_1$ denote the hypothesis of having an event, and $\mathcal{H}_o$ denote otherwise. To decide if there is an event at time $t=\tau$, we use the decision rule 
\begin{equation}
    \sum\limits_{f>BW_\text{br, adj}} \left\lvert\sum\limits_t p(t)w(t,\tau)e^{-j2\pi ft}\right\rvert^2 \,\, \underset{\mathcal{H}_0}{\overset{\mathcal{H}_1}{\gtrless}} \,\, \gamma_\text{th},
\end{equation}
where $w(t,\tau)$ is a rectangular window of length $T_\text{win}^\text{ED}$ ending at time $t=\tau$. Note that due to the time-windowing of the signal $p(t)$, the frequency spectrum of the windowed signal is that of the original signal convolved with a sinc function, which increases the bandwidth of the signal by an amount of $1/T_\text{win}^\text{ED}$. Hence, the value of $BW_\text{br}$ is adjusted to be $BW_\text{br, adj} = 2f_{o,\text{br}} + 1/T_\text{win}^\text{ED}$, where $f_{o,\text{br}}$ is the maximum normal breathing frequency.\footnote{Note that for a large window size (large $T_\text{win}$), the additional bandwidth $1/T_\text{win}$ can be neglected with respect to the original signal bandwidth. In such cases, the bandwidth calculations need not be adjusted.}

In order to determine the value of $\gamma_\text{th}$, we utilize the processed data of the calibration period (whose duration is $T_\text{cal}$) described in Sec. \ref{sec:pre-processing} to evaluate the following,
\begin{equation}
    \sigma_c^2 = \max_\tau \left\{ \sum\limits_{f>BW_\text{br,adj}} \left\lvert\sum\limits_t p_c(t)w(t,\tau)e^{-j2\pi ft}\right\rvert^2 \right\}
\end{equation}
where $p_c(t)$ is the processed data of the calibration period. $\sigma_c^2$ is then the maximum energy content of the calibration data above $BW_\text{br,adj}$, which is an estimate of the noise power in the band of $f > BW_\text{br,adj}$ when there is no event. We then set $\gamma_\text{th} = q\,\sigma_c^2$, where $q$ is a design parameter.

The zoomed-in part of Fig. \ref{fig:example_overnight} shows a sample normal sleep movement from a sleeping subject. The vertical red dashed line shows the start of the detected event using our proposed event detection module, while the vertical red dashed-dotted line shows its end. It can be seen that our event detection module was able to accurately localize the start and the end of the event.

\subsection{Event classification}\label{sec:event_classification}
Once an event has been detected, the processed data $p(t)$ during the event is then passed to an \textit{event classification} module that determines whether this event is normal or abnormal.
As discussed in Sec. \ref{sec:spectral_analysis}, the duration of a seizure is usually longer than that of any normal event. However, relying solely on event duration for deciding whether the event is a seizure or not induces an unfavorable delay in the system response, as the system would have to wait for a relatively long period of time before declaring an event as a seizure, which can lead to undesirable complications for the patient.  It is then crucial to analyze the detected events in terms of their frequency content, using the analysis and parameters derived in Sec. {\ref{sec:spectral_analysis}}, in order to have an early and robust detection. We next describe our event classification algorithm.

First, any event whose duration is less than a tolerable value $T_\text{min}$ is declared as a normal event. This step is important to avoid the unnecessary computational overhead of analyzing very short events, such as sleep jerks or very quick limb movements, since it is almost impossible for a tonic-clonic seizure to have such a short duration \cite{luders1998semiological}. It is noteworthy that this comes at the expense of a small delay in the response time, since a seizure would only be declared at least $T_\text{min}$ after its onset. As a design choice, we set $T_\text{min} = 5~$sec for our system implementation. We will show the effect of varying $T_\text{min}$ on the system performance in Sec. \ref{sec:results}. 

For the rest of the events (whose durations are larger than $T_\text{min}$), let $p_e(t)$ denote the processed CSI measurements during the event. We divide $p_e(t)$ to consecutive overlapping windows of length $T_\text{win}^\text{EC}$, and estimate the bandwidth of $p_e(t)$ as the median of the bandwidths of the signals in the overlapping windows. More specifically, the bandwidth of $p_e(t)$ is estimated as 
\begin{equation} \label{eq_median_BW}
    B_{p_e} = \underset{\tau}{\text{median}} \,\, \left\{ B : \frac{\sum\limits_{f>B} \left\lvert\sum\limits_t p_e(t)w'(t,\tau)e^{-j2\pi ft}\right\rvert^2}{\sum\limits_{f>0} \left\lvert\sum\limits_t p_e(t)w'(t,\tau)e^{-j2\pi ft}\right\rvert^2} = 0.1 \right\} ,
\end{equation}
where $w'(t,\tau)$ is a rectangular window of length $T_\text{win}^\text{EC}$ ending at $t=\tau$, and the quantity inside the braces is the 90-th percentile bandwidth of the signal within the window ending at $t=\tau$. For an ongoing long event, the bandwidth $B_{p_e}$ is updated by adding more time windows of the new data to the calculation of Eq.~\ref{eq_median_BW}. This method of estimating the bandwidth of the signal $p_e(t)$ is favorable for real-time operation, since it requires a fixed-length FFT operation for a window of size $T_\text{win}^\text{EC}$ to update the bandwidth of an ongoing event.

We declare a seizure if the bandwidth $B_{p_e}$ exceeds a threshold $f_\text{th}$.  By using the spectral analysis of Sec. \ref{sec:spectral_analysis} and the corresponding bandwidth calculations of Table \ref{table_summary_parameters}, we set $f_\text{th} = \frac{9.9+7.8}{2}=~$8.85~Hz, since this value optimally separates the bandwidths of the WiFi signal during seizures from the ones during normal events. We will study the effect of changing $f_\text{th}$ on the system performance in Sec.~\ref{sec:results}.

\section{Experimental setup} \label{sec:experimental_validation}
In this section, we describe the experimental setup we shall use as a proof-of-concept for our proposed seizure detection system.

\vspace{2pt}
\textbf{Experimental Setup:} For the WiFi CSI data collection, we use two laptops equipped with Intel 5300 WiFi cards. One of the laptops (the Tx) transmits WiFi packets at a rate of 200 packets per second on WiFi channel 48, which has a carrier frequency of 5.24~GHz. The other laptop (the Rx) uses CSItool \cite{halperin2011tool} to measure the CSI data of 30 WiFi subcarriers on 3 Rx antennas. The CSI magnitude data and the phase difference data with respect to antenna 1 (i.e. $\theta_2(t) - \theta_1(t)$ and $\theta_3(t) - \theta_1(t)$) are then logged and processed offline using MATLAB. We collect the WiFi data in 7 different dorm rooms/bedrooms (some of which are shown in Fig.~\ref{fig:locations}). In all the locations, we start by placing the Tx and Rx on two different sides of the bed on which the test subject lies down (with Tx-Rx distance of $\sim 2.5~$m). The antennas of the Tx and Rx are both elevated by 70~cm above the bed level. We then study the impact of different Tx/Rx configurations in Sec. \ref{sec:results}. Note that for the Rx, external tripod-mounted antennas may be used in order to make the Rx at the same height as the Tx. This configuration for the relative positioning of the Tx, the Rx, and the bed results in $\psi \approx 1$ (the angle $\phi$ in Fig. {\ref{fig:scenario}} is $\sim 60^\circ$), independent of the person's pose or orientation on the bed.

\vspace{2pt}
\textbf{Test Subjects and Experiment Protocol:} We recruited a total of 20 student actors (5 females and 15 males) to participate in our experiments, where each subject participates in one or more of the experimental locations. In total, the number of subjects participating in each of the 7 locations are 11, 6, 4, 2, 1, 1, and 1 subjects, respectively.\footnote{The Institutional Review Board (IRB) committee has reviewed this research and determined that it does not constitute human subject research. Furthermore, all the experiments that were carried out during the pandemic followed the strict COVID-19 safety guidelines put in place by our institution.} 
Each participant was consensually trained on how to simulate a tonic-clonic seizure and shown public online YouTube videos explaining how tonic-clonic seizures look like. It is worth noting that seizure acting is a common practice in medical schools, where healthy persons (known as \textit{standardized patients}) are recruited to act out different medical conditions to provide introductory training opportunities for medical students \cite{barrows1993overview,dworetzky2015interprofessional}. Hence, testing a system on simulated seizures is an important step towards more advanced clinical  trials. 

For each subject, the receiver starts logging the CSI data when the subject is in a sleep state (only breathing and in any generic position) for at least 15 seconds (part of which to be used as one-time calibration data). Then the subject starts simulating seizures and normal sleep movements. Each seizure instance is simulated for at least 20 seconds.\footnote{An actual tonic-clonic seizure can last for 1 to 3 minutes. However, it is a physically-challenging task for a healthy person to simulate it for such a long time.} In total, each participant does 10 seizure simulations per location, resulting in a total of 260 independent instances of seizure data across all the locations. Similarly, each subject performs several normal non-breathing sleep movements spontaneously in each experiment. By observing the subjects' movements, they included posture adjustments (e.g. switching from lying on their side to lying on their back), limb-only movements (e.g. stretching or tucking the knee), scratching, stretching, coughing, sneezing, and sleep jerks. Overall, we collected a total of 410 independent normal sleep events from all subjects in all the locations.\footnote{Sample data files and detection/classification codes are available in this URL \url{https://doi.org/10.21229/M9ZT09}}

\vspace{2pt}
\textbf{Performance Metrics}:
We test the performance of our system according to the following two performance metrics:

\noindent 1. Seizure Detection Rate (SDR): which is defined as the number of detected seizures, divided by the total number of seizures (expressed as a percentage).

\noindent2. Probability of False Alarm ($P_\text{FA}$): which is defined as the number of normal sleep events which are incorrectly classified as seizures, divided by the total number of detected normal sleep events.

\noindent3. Response Time (RT) for seizures: which is defined as the time at which the event classification module detects the seizure, measured with respect to the seizure's onset.

\vspace{2pt}
\textbf{Algorithm Parameter Values:} 
We set the following values for different algorithm parameters, $T_\text{cal}=~$13~sec, $T_\text{win}^\text{ED}=~$2~sec, $T_\text{win}^\text{EC}=~$4~sec, $T_\text{min}=~$5~sec, $K=~$15, and $q=~$2. \rev{The optimum classification threshold, $f_\text{th}$, is then found based on our proposed mathematical framework to be $f_\text{th}=8.85$~Hz, as shown in Sec.~\ref{sec:spectral_analysis}.  It is worth stressing that this threshold is found based on our rigorous theoretical characterization of the bandwidth, and not based on empirical data.} The effect of varying some of these parameters on system performance is shown in Sec. \ref{sec:results}.

\begin{figure}
\centering
  \includegraphics[width=1\linewidth]{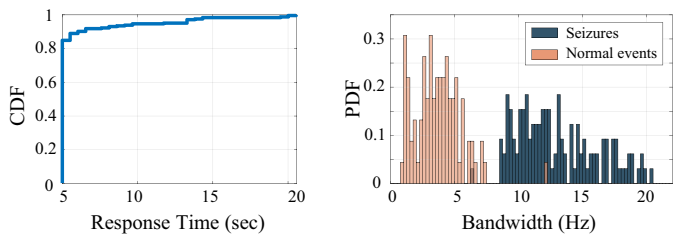}
  \caption{\small (Left) CDF of the system's response time. The mean response time is 5.69 seconds.
  (Right) PDF of the bandwidths of seizure events and normal sleep events, showing a gap in the 7-9~Hz band.}
  \label{fig:misc_results}
\end{figure}

\section{Evaluation Results}\label{sec:results}
In this section, we present the performance evaluation results of our proposed seizure detection algorithm.

For the seizure data instances in all the locations, our proposed system was able to detect 244 out of the 260 seizures, resulting in a seizure detection rate (SDR) of 93.85\%. It is worth noting that the event detection module was able to detect all the 260 seizures. However, the event classification module misclassified 16 out of the detected 260 events. Fig. \ref{fig:misc_results}~(left) shows the CDF of the response times of the detected seizures, showing that our system achieves a Mean Response Time (MRT) of 5.69 sec.
Such an early detection is important for the caregiver to provide the needed medical assistance as soon as possible. In terms of locations, the seizure detection rate in the 7 locations was 93.6\%, 95\%, 90\%, 90\%, 100\%, 100\%, and 100\%, while the mean response time in the 7 locations was 5.8, 5.8, 5.68, 5.58, 5.65, 5 and 5.1 sec, respectively. This shows that the system's performance is insensitive to the deployment environment, since the static multipath does not affect the information-bearing parts of the received WiFi signal, as discussed in Remark 2.

In terms of normal events, our event detection module was able to detect 406 out of the 410 normal events. It is worth noting that it is irrelevant if the system misses some normal events, as the main purpose of the system is seizure detection with as few false alarms as possible.  
Out of the detected normal events, only 4 events were incorrectly classified as seizures, resulting in a probability of false alarm $P_\text{FA} =~$0.0097. Fig.~\ref{fig:misc_results}~(right) shows the densities of the measured bandwidths of the WiFi signals during seizure events as well as normal sleep events.  
The distributions of the bandwidths show a clear gap in the band of 7-9~Hz, which validates the theoretical bandwidth characterization of Sec. \ref{sec:spectral_analysis}.

\begin{figure}
\centering
  \includegraphics[width=1\linewidth]{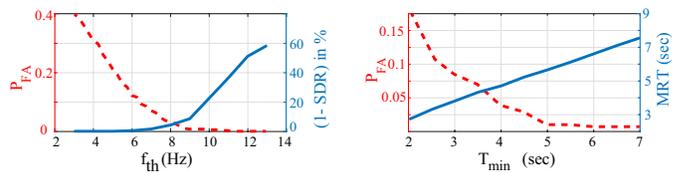}
  \caption{\small (Left)  $P_\text{FA}$ and (1-SDR) as a function of $f_\text{th}$: increasing $f_\text{th}$ degrades SDR while improving $P_\text{FA}$. (Right) MRT and $P_\text{FA}$ as a function of $T_\text{min}$: increasing $T_\text{min}$ degrades MRT while improving $P_\text{FA}$.}
  \label{fig:fth_tmin_effect}
\end{figure}

\vspace{3pt}
\noindent\textbf{Processing time:} It takes 18~ms, on average, to process one second of collected data, using our algorithm of Sec.~\ref{sec:system_description}.

\vspace{3pt}
\noindent\textbf{Comparison to state-of-the-art:} {\cite{van2016non}} provides a survey for in-home tonic-clonic seizure detection algorithms that use different modalities, e.g. accelerometry, mattress units, and video, to detect tonic clonic seizures on real epilepsy patients. Table {\ref{tab_results_comparison}} compares the performance of our proposed system to  the performance of the different detection techniques reported in the survey of {\cite{van2016non}}, as well as other multimodal seizure detection papers.
Overall, our results show the robustness of our proposed system, in terms of achieving a very good seizure detection rate, probability of false alarm, and a fast average response time of 5.69 seconds to detect a seizure, while being non-invasive, and privacy preserving. We note that part of the contribution of this paper was also to develop a new mathematical model that can enable seizure detection using WiFi signals, while most of the existing work is mainly either testing an existing product, or utilizing straightforward modalities, e.g. accelerometry. Furthermore, our approach is the only privacy-preserving one that is also non-invasive. While our results are based on simulated seizures, they constitute a strong proof-of-concept for our proposed idea/mathematical models, which shows how RF signals (e.g. WiFi) can be used as a non-invasive, robust, and affordable alternative for nocturnal seizure detection. Hence, our proposed algorithms serve as a basis for a system that can be subsequently tested in clinical settings, towards the ultimate goal of making such technology available to the public.

\subsection{Effect of varying $f_\text{th}$ and $T_\text{min}$}
Based on our theoretical analysis of Sec. \ref{sec:spectral_analysis},  we concluded that a threshold of $f_\text{th} =~$8.85~Hz optimally separates the bandwidth of the WiFi signals during normal sleep movements from that during seizures. In this section, we study the effect of varying $f_\text{th}$, while keeping all other system parameters at their default values. Fig.~\ref{fig:fth_tmin_effect}~(left) shows (1-SDR) and $P_\text{FA}$ as a function of $f_\text{th}$. It can be seen that SDR decreases (becomes worse) when increasing $f_\text{th}$, since more seizure events can go undetected due to their bandwidth being less than the higher $f_\text{th}$. On the other hand, increasing $f_\text{th}$ improves $P_\text{FA}$, since it becomes less likely for the bandwidth of the WiFi signal during a normal sleep event to exceed a higher $f_\text{th}$. We can see that the mathematically-driven value of 8.85~Hz strikes a good balance between SDR and $P_\text{FA}$.

Next, we study the effect of varying $T_\text{min}$, which is the minimum duration for an event to be passed to the event classification module.  Fig.~\ref{fig:fth_tmin_effect}~(right) shows MRT and $P_\text{FA}$ as a function of $T_\text{min}$.  Expectedly, increasing $T_\text{min}$ increases MRT, since a higher $T_\text{min}$ means that the event classification module (which determines whether the event is a seizure or not) is not activated for a longer time after the seizure onset. On the other hand, increasing $T_\text{min}$ improves $P_\text{FA}$, since a higher portion of the normal events are declared normal by default due to their short duration. It can be seen that the chosen value of $T_\text{min}~$= 5 sec strikes a good balance in the MRT-$P_\text{FA}$ tradeoff. It is worth noting that SDR does not change as a function of $T_\text{min}$ in Fig.~\ref{fig:fth_tmin_effect}, and as such is not plotted.

\begin{table}
\rowcolors{2}{white}{gray!09}
	\centering
	\caption{\small Comparison with state-of-the-art in seizure detection. }
	\label{tab_results_comparison}
    \scriptsize
	\begin{tabular}[h]{ c | c | c | c | c | c | c}
\Xhline{2\arrayrulewidth}
\darkgray
		  &  & \textbf{Seizure} & \textbf{MRT} &  \textbf{$P_\text{FA}$} & \textbf{Non-} & \\ 
\darkgray
		\multirow{-2}{*}{\textbf{Paper}} & \multirow{-2}{*}{\textbf{Modality}} & \textbf{det. rate} & \textbf{(sec)}	&\textbf{/night } & \textbf{invasive} & \multirow{-2}{*}{\textbf{Privacy}}\\ 
\Xhline{2\arrayrulewidth}

\rowcolor[gray]{.96}		Avg. of   &  & &  &  &  &  \\
    	\cite{van2016non} & \multirow{-2}{*}{Acc.} & \multirow{-2}{*}{90.7\%} & \multirow{-2}{*}{41.1} & \multirow{-2}{*}{0.3} & \multirow{-2}{*}{\xmark} & \multirow{-2}{*}{\checkmark} \\
		\cite{kramer2011novel} & Acc. & 91\% & 17 & 0.1 & \xmark & \checkmark \\
		\cite{poppel2013prospective} & \rev{Piezo Tech.}	& 84.6\%       & -- & --  & \xmark & \checkmark  \\ 
		\cite{kalitzin2012automatic} & Video	& 95\%       & -- & 1  & \checkmark & \xmark  \\ 
		\cite{arends2018multimodal} & Acc.+HR	& 96\%       & 15 & 0.23  & \xmark & \checkmark  \\ 

\Xhline{1\arrayrulewidth}
\darkgray
		\textbf{Ours}	    & \textbf{WiFi} & \textbf{93.85\%}    &\textbf{5.69}    &\textbf{0.23}*  & \checkmark & \checkmark  \\ 
\Xhline{2\arrayrulewidth}
	\end{tabular}
	
	\vspace{3pt}
	\scriptsize { * Based on an average of 3 normal events per hour, for 8 hours of night sleep. \\Acc = Accelerometry. HR = Heart Rate. \rev{Piezo technology is used in the form of units installed in/under mattresses to detect pressure changes.}}
\end{table}

\subsection{Effect of Tx-Rx positioning} \label{sec:psi}
For all the previous results, we considered a setting (henceforth denoted by C\#1) where the position of the Tx/Rx with respect to the bed resulted in a value of $\psi \approx 1$ (see Fig. \ref{fig:locations}). In this section, we test our system with different Tx/Rx positions resulting in different $\psi$s. Based on our analysis of Sec.~\ref{sec:signal_model}, $\psi$ impacts the bandwidth characterization of seizure and normal sleep movements as follows:
$BW=\frac{\psi v_\text{max}}{\lambda}+f_o$, where $v_\text{max}$ and $f_o$ denote the $v_\text{max}$ and $f_o$ of the corresponding cases.

To test the sensitivity of our system to different Tx/Rx positions and their corresponding $\psi$, we carry out extensive experiments on one test subject (in location 4 of Fig. \ref{fig:locations}) by changing either the Tx/Rx locations in the same horizontal plane (to which we refer as changing their configuration), or changing their heights.    

\vspace{3pt}
\noindent\textbf{Changing Tx/Rx configuration:} In order to test the sensitivity of the system to the placement of the Tx/Rx in the horizontal plane, we conduct experiments in two additional configurations,  C\#2 and C\#3. In C\#2, the Tx and the Rx are placed on one side of the bed, such that the line connecting the Tx and Rx is parallel to the edge of the bed and 70~cm away from it. Such a configuration can be of particular interest in practical situations in which one side of the bed is not accessible, e.g. if the bed is placed next to a wall. The distance between the Tx and the Rx is 2~m, and both are elevated by 70~cm above the bed level. This setup results in $\psi\approx1.4$, which will result in $BW_\text{sz}\ge~$13.23~Hz, and $BW_\text{nm} \le~$10.06~Hz using our mathematical derivations. We thus use $f_\text{th}=\frac{13.23+10.06}{2}=~$11.64~Hz for this configuration.
On the other hand, in C\#3, the Tx and Rx are placed on two different sides of the bed, with a Tx-Rx distance of 3.6~m, while they are elevated by 70~cm above the bed level. This setup results in $\psi\approx~$0.7, and $f_\text{th}=\frac{7.36+6.03}{2}=~$6.69~Hz.

\begin{table}
    \centering
    \caption{\small Performance in different Tx/Rx placement settings. }
	\label{tab_config_psi}
 	\scriptsize
	\begin{tabular}{ >{\centering\arraybackslash}p{1cm} |  >{\centering\arraybackslash}p{.4cm} |  >{\centering\arraybackslash}p{0.9cm} |  >{\centering\arraybackslash}p{0.6cm} |  >{\centering\arraybackslash}p{.6cm} |  >{\centering\arraybackslash}p{0.8cm}  >{\centering\arraybackslash}p{1.2cm}}
\HHHline{1pt}{-|-|-|-|-|-|~}

\darkgray
		 &  &  & \textbf{SDR} & \textbf{MRT}&\\ 
		
\darkgray
		\multirow{-2}{*}{\textbf{Setting}} & \multirow{-2}{*}{\textbf{$\psi$}} & \multirow{-2}{*}{\textbf{$f_\text{th}$ (Hz)}} & \textbf{\% }	&\textbf{(sec)} & \multirow{-2}{*}{\textbf{$P_\text{FA}$}}\\ 
		
        % \textbf{Setting} & \textbf{$\psi$} &  & \textbf{MRT} & \textbf{$P_\text{FA}$} \\
\HHHline{.8pt}{-|-|-|-|-|-|~}
\lightgray		C\#1 & 1    & 8.85    & 93.85 & 5.69   & 0.0097   \\ 
		\HHHline{.3pt}{======~}
\lighttgray		$\,\,\,$C\#2\newline C\#3 & $\,1.4$ 
		\newline 0.7	& $\,11.64$ \newline 6.69 & $\,\,90$ \newline 100  & $\,5.67$ \newline 6.75  & 0.008\newline 0.016  & \cellcolor{white}\hspace{-4pt}\rdelim\}{2}{*}{\footnotesize \textit{Different configs.}}\\ 
		\HHHline{.3pt}{======~}
\lightgray        $\,\,\,$C\#4 \newline C\#5 & 1.44 \newline 1.61	& $\,11.94$ \newline 13.15 & $\,\,90$ \newline 100 & $\,\,5$  \newline 7.3 & 0.008 \newline 0.016 & \cellcolor{white}\hspace{-4pt}\rdelim\}{2}{*}{\footnotesize \textit{Different heights}} \\
%\Xhline{2\arrayrulewidth}
\HHHline{.8pt}{-|-|-|-|-|-|~}
	\end{tabular}
	
\end{table}

In each of the configurations, the test subject simulates a total of 10 seizure instances and 125 normal sleep events. We summarize the evaluation results of these experiments in Table \ref{tab_config_psi}. It can be seen that the performance of the system in C\#2 and C\#3 is comparable to that of the main configuration (C\#1), showing that the performance of our proposed pipeline is robust to different Tx/Rx configurations.

\vspace{3pt}
\noindent\textbf{Changing Tx/Rx heights:} In order to test the sensitivity of the system to antenna heights, we conduct experiments in two additional settings, C\#4 and C\#5. In both settings, the Tx and Rx are placed $\sim$2.5~m apart on both sides of the bed (similar to C\#1), but their heights are elevated to 1.3~m above the bed level in C\#4, and 1.7~m above the bed level in C\#5. Using simple geometry, it can be verified that in C\#4, $\psi=$~1.44 ($f_\text{th}=~$11.94~Hz), while in C\#5, $\psi=~$1.61 ($f_\text{th}=~$13.15~Hz). Again, in both settings, the test subject simulates a total of 10 seizure instances and 125 normal sleep events. Table \ref{tab_config_psi} shows that the performance of the system in C\#4 and C\#5 is comparable to that of the other configurations, indicating that our proposed pipeline is robust to different Tx/Rx heights.

\subsection{Multi-person operation} \label{sec:multipeople_results}
In-home seizure detection systems are primarily designed for caregivers who do not share the same bed (or bedroom) as the patient, since, otherwise, they would be alerted by the patient's seizure movements. However, in order to show the robustness of our proposed system, we next show that it can still be deployed in a multi-person setting where multiple people share the same bed. In such a case, the event detection module detects any movement done by any of the sleeping persons. In order to test this, we conducted a 10-minute experiment where two people lie down next to each other on a bed. Person~1 simulates seizures at the 6 and 9-minute marks. Otherwise, both people frequently simulate normal sleep movements. As such, there are a number of instances where both people move at the same time, or one person moves normally while the other one is simulating a seizure. Fig.~\ref{fig:example_multipeople}~(top) shows the PCA-denoised stream $p(t)$, in which perturbations are clearly visible whenever either persons engages in any kind of movement, while  Fig.~\ref{fig:example_multipeople}~(bottom) shows the bandwidth of the WiFi signals during the detected events. It can be seen that the bandwidth exceeds $f_\text{th}$ for an extended period of time only during the seizure instances, which are correctly classified as seizures, even though the second person was moving during the second seizure instance. Otherwise, during normal movements, even if both persons are moving simultaneously, the events are not classified as seizures.

\section{Further Discussions}\label{sec:discussion}
\noindent\textbf{Robustness to movements by other people:} In Sec.~\ref{sec:multipeople_results},
we have shown that our proposed system is robust to movements by other sleeping people in the same environment, since their normal sleeping movements have the same characteristics as those of the patient. Next, consider the case where other simultaneous movements happen, such as those of a walking person.  The spectrogram of the signal reflected off of a walking person has specific characteristics.  Thus, as part of future work, one can study the differentiability of the signals induced by walking from those induced by seizures.  Furthermore, the reflected signals off of other moving targets can also be filtered out at the Rx by exploiting more signal dimensions. For instance, multiple antennas at the Rx can separate the received signals based on their Angle-of-Arrival (AoA).

\begin{figure}
    \centering
    \includegraphics[width=1\linewidth]{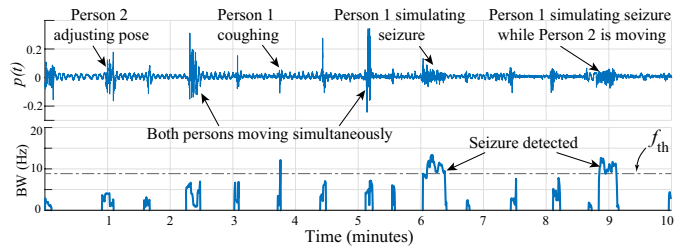}
    \caption{\small (Top) The PCA-denoised data $p(t)$ in a 10-minute experiment with 2 subjects. (Bottom) The bandwidth of $p(t)$ during the detected events. It can be seen that the seizures are the only events whose bandwidth exceeds $f_\text{th}$ for an extended period of time.}
    \label{fig:example_multipeople}
\end{figure}

%\vspace{4pt}
%\noindent\textbf{Real-time processing:} While the results presented in this paper are based on the offline processing of the collected CSI data, the design of the algorithm has taken real-time processing considerations into account, by relying on moving windows. Moreover, our proposed algorithm is computationally efficient, taking only 18~ms on average to process one second of data. 
%While the results presented in this paper are based on the offline processing of the collected CSI data, the design of the algorithm has taken real-time processing considerations into account. For instance, the steps of the algorithm, such as event detection and bandwidth estimation, rely on moving windows, instead of processing all the data at once, which facilitates the implementation of the algorithm in real-time. Moreover, our proposed algorithm is computationally efficient, taking only 18~ms on average to process one second of the data. Since the window size of the event classification module is 4~sec, this means that the additional processing delay to the system's response time (whose mean is 5.69~sec) is only 18$\times$4=72~ms. 

\vspace{4pt}
\noindent\textbf{Clinical trials:} 
In this paper, we proposed the first RF-based system for nocturnal seizure detection, by developing mathematical models that can enable this. We also validated our proposed approach by extensive experiments on seizures simulated by actors. The results of this preliminary validation show the great potential of using WiFi signals as an appealing alternative to the currently available products, which are costly, uncomfortable, or unreliable. Towards the ultimate goal of making this technology available to the public, the next step is to develop a prototype of the proposed system, which can then undergo extensive clinical trials on real patients, and become available to the epilepsy patients and their caregivers.

\section{Conclusion}\label{sec:conclusion}
In this paper, we have considered the problem of nocturnal seizure
detection in epilepsy patients using WiFi signals measured on a device placed in the vicinity of the sleeping patient. We first provided a  mathematical analysis for the spectral content/bandwidth of the WiFi signal during different kinds of sleep body movements (e.g., seizure, normal movements, and breathing), showing that the bandwidth of the signal can be used to robustly differentiate a seizure from normal movements. We then utilized this analysis to design a robust seizure detection system, which detects all non-breathing body motion events and classifies them, based on their spectral content, to normal movements and seizures. We experimentally validated our proposed system using WiFi CSI data collected from 20 actors in 7 different locations, where they simulated a total of 260 seizures as well as 410 normal sleep movements. Our proposed system achieved a very low probability of false alarm of 0.0097, while being very responsive to seizure events, detecting 93.85\% of the seizure instances with an average response time of only 5.69 seconds. 
These promising results show the potential of using WiFi signals as an accurate and cheap alternative to traditional seizure detection systems.

\bibliographystyle{IEEEtran}

\bibliography{ref_seizure}

\end{document}